\documentclass[prl,letterpaper,aps,superscriptaddress,floatfix,twocolumn,nofootinbib]{revtex4-2}
\pdfoutput=1
\usepackage{graphicx}
\usepackage{amsmath}
\usepackage{amsfonts}
\usepackage{amssymb}
\usepackage[normalem]{ulem}
\usepackage[small,bf]{subfigure}
\usepackage[utf8]{inputenc}
\usepackage{dsfont}
\usepackage{xcolor}
\usepackage{slashed}
\usepackage{soul}
\usepackage{comment}
\usepackage{tikz-cd}
\usepackage{hyperref}

\newcommand{\beq}{\begin{equation}}
\newcommand{\eeq}{\end{equation}}

\newcommand{\beqa}{\begin{eqnarray}}
\newcommand{\eeqa}{\end{eqnarray}}

\newcommand{\cH}{{\cal H}}

\newcommand{\textoverline}[1]{$\overline{\mbox{#1}}$}

\newcommand{\ea}{\end{array}}

\def\eea{\end{eqnarray}}

\def\<{\langle}
\def\>{\rangle}

\def\bZ{\mathbb{Z}}

\def\cO{\mathcal{O}}

\usepackage{amsmath}
\usepackage{comment}
\usepackage{amssymb}
\usepackage{amsthm}
\newtheorem{thm}{Theorem}

\newtheorem{prop}{Proposition}
\newtheorem{conj}{Conjecture}
\newtheorem{cor}{Corrollary}
\theoremstyle{definition}

\usepackage{color}
\usepackage{tikz-cd}
\usepackage{comment}

\newcommand{\ket}[1]{|#1\rangle}
\newcommand{\bra}[1]{\langle#1|}

\usepackage{tikz}
\usepackage[export]{adjustbox}

\def\[#1\]{%
  \begin{equation}\begin{gathered}#1\end{gathered}\end{equation}%
}

\usepackage[safe]{tipa}

\begin{document}
\title{$W$ state is not the unique ground state of any local Hamiltonian}
\author{Lei Gioia}
 \affiliation{Walter Burke Institute for Theoretical Physics, Caltech, Pasadena, CA, USA}
 \affiliation{Department of Physics, Caltech, Pasadena, CA, USA}
\affiliation{Kavli Institute for Theoretical Physics, Santa Barbara, California 93106, USA} 
\author{Ryan Thorngren}
\affiliation{Mani L. Bhaumik Institute for Theoretical Physics, Department of Physics and Astronomy,
University of California, Los Angeles, CA 90095, USA}
\affiliation{Kavli Institute for Theoretical Physics, Santa Barbara, California 93106, USA}
\affiliation{Institut des Hautes Études Scientifiques, 91440 Bures-sur-Yvette, France}

\begin{abstract}
The characterization of ground states among all quantum states is an important problem in quantum many-body physics. For example, the celebrated entanglement area law for gapped Hamiltonians has allowed for efficient simulation of 1d and some 2d quantum systems using matrix product states. Among ground states, some types, such as cat states (like the GHZ state) or topologically ordered states, can only appear alongside their degenerate partners, as is understood from the theory of spontaneous symmetry breaking. In this work, we introduce a new class of simple states, including the $W$ state, that can only occur as a ground state alongside an \textit{exactly} degenerate partner, even in gapless or disordered models. We show that these states are never an element of a stable gapped ground state manifold, which may provide a new method to discard a wide range of `unstable' entanglement area law states in the numerical search of gapped phases. On the other hand when these degenerate states are the ground states of gapless systems they possess an excitation spectrum with $O(1/L^2)$ finite-size splitting. One familiar situation where this special kind of gaplessness occurs is at a Lifshitz transition due to a zero mode; a potential quantum state signature of such a critical point. We explore pathological parent Hamiltonians, and discuss generalizations to higher dimensions, other related states, and implications for understanding thermodynamic limits of many-body quantum systems.
\end{abstract}
\maketitle

On a length $L$ 1d chain of spin-$\frac12$'s the $W$ state is defined as 
$\ket{W_1}=\frac{1}{\sqrt{L}}\left(\ket{10...0}+\ket{010...0}+...+\ket{0...01}\right)$.
This state has received much attention in the literature \cite{joo2003quantum,frowis2013certifiability,Faist_2020,lightmatterW,RevModPhys.93.045003,10.21468/SciPostPhys.15.4.131}, as it is representative of a class of states distinct both from short-range entangled states and from familiar long-range entangled states such as macroscopic superpositions (e.g. GHZ states) and topologically ordered states (e.g. toric code states), for instance in the entanglement classification \cite{D_r_2000}. It has also become a target for state-preparation protocols on existing and near-term quantum hardware \cite{zang2015generating,Cruz_2019,heo2019scheme}. In this letter, we provide a sharp characterization in the form of a no-go theorem, which roughly says the $W$ state (and its relatives, we dub \textit{wh\textoverline{a}nau-states}) is a ground state of a local Hamiltonian only if the all-zero state $\ket{0}$ is as well. This provides a barrier for adiabatic preparation of the $W$ state despite its low area-law entanglement entropy, existence of finite bond dimension matrix-product state description~\cite{mpspaper,RevModPhys.93.045003}, and general `simplicity' of the state. We also describe some condensed matter implications of the theorem, which turns out to be deeply intertwined with the physics of quantum Lifshitz transitions \cite{voloviklifshitz}, and challenges our current concepts of the thermodynamic limit.

Our results apply to a very broad class of Hamiltonians which includes the usual translation-invariant (with any unit cell) or disordered Hamiltonians studied in condensed matter, as well as more exotic Hamiltonians with domain walls and other defects inserted in particular ways. However, it rules out some pathological Hamiltonians we consider later in the paper, such as those whose coefficients depend explicitly on $L$ or all-to-all Hamiltonians.
Specifically, we define a (infinite or half-infinite, periodic or open) \emph{finite-range 1d Hamiltonian system} to be a sequence of Hilbert spaces $\cH_k$ and a sequence of bounded norm operators $h_k$ acting on $\bigotimes_{j = k-l}^{k+l} \cH_j$ (for some fixed range $l$), where either $k \in \bZ$ in the infinite case or $k \in \bZ_{> 0}$ in the half-infinite case. For each system size $L$, these define Hamiltonians $H_L = \sum_{k= -\lfloor (L-1)/2 \rfloor}^{\lceil (L-1)/2 \rceil} $ or $H_L = \sum_{k=1}^L h_k,$
on each length $L$ chain $\cH(L) = \bigotimes_k \cH_k$ (with ranges as above), with either periodic or open boundary conditions, where in the periodic case, we let the action of $h_k$ within $l$ of the edges of the range wrap around the chain. Note that the periodic case requires some identification between the local Hilbert spaces $\cH_k$.

Consider a length $L$ 1d chain of spin-$\frac12$'s. Starting with the ``all-zero'' state $\ket{0}$, defined by $Z_j \ket{0} = \ket{0}$ for all $j$, for each $n$, we can further define the generalization of the $W$ state to be
\[\ket{W_n} = {L \choose n}^{-1/2} \sum_{i_1 < \cdots < i_n} X_{i_1} \cdots X_{i_n} \ket{0},\]
where $\ket{W_1}$ is simply the $W$ state.
State $\ket{\psi_L}$ defined on a sequence of length $L$ Hilbert spaces as above, such as $\ket{W_n}$, is a ground state of a local Hamiltonian system if for some large enough sequence of $L$, $\ket{\psi_L}$ is a lowest energy eigenstate of $H_L$. With these definitions, we can state our main result:

\begin{thm}
\label{thm:Wstate}
    The $W$ state (and more generally $\ket{W_n}$) is not the unique ground state of any finite-range 1d Hamiltonian system. In particular, if it is a ground state, $\ket{0}$ is also an exactly degenerate ground state.
\end{thm}
The idea of the proof is that any negative energy the $W$ state has relative to $\ket{0}$ must be associated with the single 1. When we study $\ket{W_2}$, it has two 1s, which are likely far apart and each contribute the same negative energy, so $\ket{W_2}$ must therefore have even less energy than the $W$ state, meaning the $W$ state could not have been the unique ground state.

The previous best result is that the gap above the $W$ state scales as $O(1/L^{3/2})$ \cite{macroscopicsuperpositions}. We show that beyond the \emph{exact} degeneracy with $\ket{0}$, by studying the ``boosted $W$ states'' we find an $O(1/L^2)$ excitation spectrum.

\section{Proof of the No-Go Theorem}

Let us now give the proof of the no-go theorem regarding the $\ket{W_n}$ states. We will give two propositions that capture what is really special about the $W_n$ states and then synthesize them into a proof of the theorem. First, we will show that in the large $L$ limit, the expected energy differences between $\ket{W_n}$ and $\ket{W_{n+1}}$ are equal for all $n$. In particular, we have the following:

\begin{prop}
Consider a finite-range 1d Hamiltonian system with Hamiltonians $H_L$. (We do not assume $\ket{0}$ or $\ket{W_n}$ are eigenstates.) Define
\[\Delta_1 := \langle 0 |H_L | 0 \rangle - \langle W_1 | H_L | W_1 \rangle,\]
then for all $n>1$,
\[\Delta_n := \langle 0 |H_L| 0 \rangle - \langle W_n | H_L | W_n \rangle = n\Delta_1 + O(1/L).\]
\end{prop}

The idea of the proof is to write $\Delta_1$ as a sum of contributions from where the $X_i$'s are inserted in $\ket{0}$. Then, since the expected energy is local, we only get nonzero contributions when $h_k$ is near where the $X_i$'s are inserted. Then, when we compute $\Delta_n$, most of the contribution is from when the $X_i$ are distantly separated, so each will contribute independently a factor of $\Delta_1$, up to errors going to zero with $L$. The proof is detailed in the Supplemental material~\cite{supp}.

Another very special property of the $W$-state is:
\begin{prop}
    If $\ket{W_n}$ is an eigenstate, then $\Delta_n$ is independent of $L$ for $L > n + 2l$.
\end{prop}

\begin{proof}
    Let us first demonstrate the case for $n = 1$.

    It is convenient to write the Hamiltonian in normal ordered form
    \[ H_L = \sum_k :h_k: + C_k,\]
    where $\langle 0 | :h_k: | 0\rangle = 0$. We are free to set $C_k = 0$ since we are only interested in energy differences. Then $H_L \ket{W_1} = \Delta_1 \ket{W_1}$. Let us focus on one term $\frac{\Delta_1}{\sqrt{L}}X_i \ket{0}$ which must appear in $H_L \ket{W_1}$. For $j$ with $|j-i| > 2l+1$,
    \[\label{eqnXvanishing}\bra{0} X_j H_L X_i \ket{0} = 0,\]
    (this is the special property of $\ket{W_1}$) so this term must come from $H_L$ applied to
    \[\frac{1}{\sqrt{L}} \sum_{j = i-l}^{i+l} X_j \ket{0}.\]
    We can isolate it by forming the matrix element (which also takes care of the normalization), giving
    \[\Delta_1 = \bra{0} X_i H_L \sum_{j = i-1}^{i+1} X_j \ket{0}.\]
    Finally, since $H$ is composed of finite range, normal ordered terms, we can discard pieces that cannot connect $i$ and $j$, so
    \[\Delta_1 = \bra{0} X_i \sum_{k = i-l}^{i+l} \sum_{j = i-l}^{i+l} :h_k: X_j \ket{0},\]
    which is manifestly independent of $L$ when $L > 2l+1$.

    To prove the case for general $n$, we look instead at a particular configuration of 1's, such as when they are all next to each other, by studying the piece $X_{i} X_{i+1} \cdots X_{i+n-1} \ket{0}$ appearing in $\ket{W_n}$. The proof goes through as above.
\end{proof}

Combining these two propositions, we can finish the proof of Theorem 1:

\begin{proof}
    Suppose towards a contradiction that for some $n > 0$, and all large enough $L$, $\ket{W_n}$ is the lowest energy eigenstate. This means $\Delta_n$ is eventually greater than zero, and in fact by proposition 2 it is eventually a positive constant. By proposition 1, we thus see $\Delta_1 = \frac{1}{n} \Delta_n + O(1/L)$ is eventually positive. Using proposition 1 again,
    \[\bra{W_n} H_L \ket{W_n} - \bra{W_{n+1}} H_L \ket{W_{n+1}} = \Delta_{n+1} - \Delta_n \\ = \Delta_1 + O(1/L)\]
    is also eventually positive, so $\ket{W_{n+1}}$ has even lower energy than $\ket{W_n}$, a contradiction!

    Therefore, if $\ket{W_n}$ is a ground state, $\Delta_n = 0$, so $\ket{0}$ is also a ground state.

\end{proof}

The argument above may seem in contradiction with the existence of \emph{any} ground state of $H$ if $\Delta_1 > 0$. However, there is a trick with the order of limits. The error in the formula for $\Delta_n$ is $O(1/L)$ (unless they are eigenstates) but also grows linearly with $n$, so the true ground state in the thermodynamic limit is one with some nonzero ``charge density'' $n/L$, which is what we physically expect. We return to this point below.

Also, it is very important in Proposition 2 that $\ket{W_n}$ is assumed to be an eigenstate. Only this way can we compute $\Delta_n$ locally. When it is not an eigenstate, $\Delta_n$ requires an average over the whole system, which introduces $L$ dependence. Without it, we would conclude all $\ket{W_n}$'s must be degenerate, however we will demonstrate below a Hamiltonian with only $\ket{0}$ and the $W$-state as its two ground states.

If we relax the finite-range condition, and just ask that the support of the terms $h_k$ fall off with some prescribed decay, such as exponential or power law, we expect our results will still hold, up to finite-size splittings with the same decay. Higher dimensional generalizations apply straightforwardly to states such as $\frac{1}{L^{d/2}} \sum_i X_i \ket{0}$, where $i$ ranges over a $d$-dimensional lattice.

\section{Consequences}

In this Section we present some immediate consequences of Theorem~\ref{thm:Wstate} for both gapless and gapped Hamiltonians.

\subsection{Gapless Hamiltonians \& Lifshitz transitions}

\begin{figure*}
    \centering
   \includegraphics[width=0.8\textwidth]{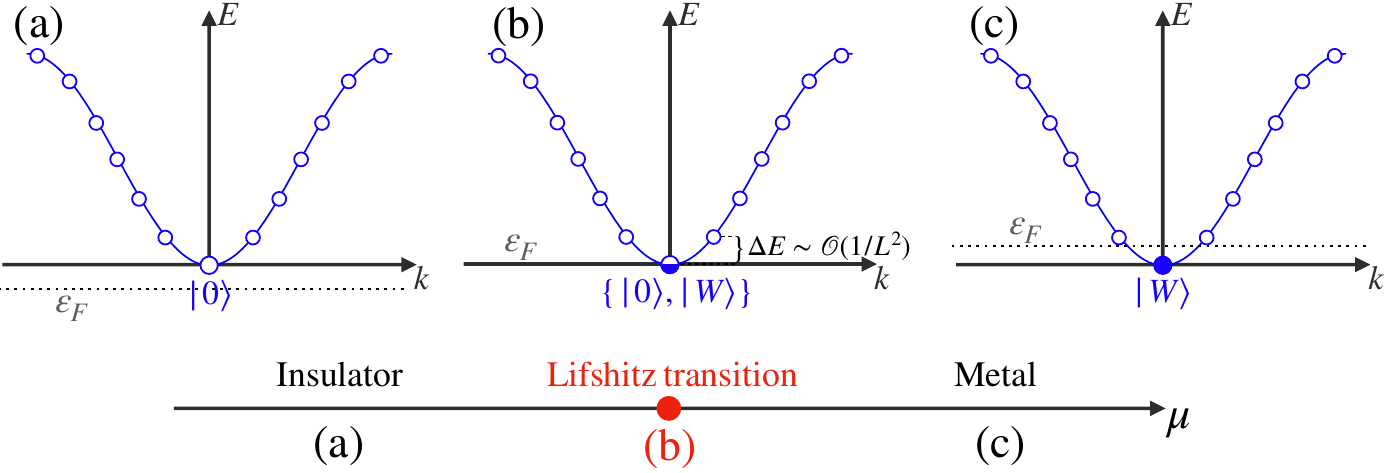}
    \caption[]{ A simple Lifshitz transition is depicted. In (a) we have a fully empty band since the Fermi energy $\epsilon_F$ is below the band - this corresponds to the insulating phase. As we increase the chemical potential $\mu$ we arrive at a critical Lifshitz transition as depicted in (b). Here the ground state manifold contains both the $\ket{0}$ and $\ket{W}$ states, a hallmark of a Lifshitz transition. In (c) we increase the chemical potential further such that the Fermi energy lies within $O(1/L^2)$ above the $\ket{W}$ state we have the $\ket{W}$ state as the ground state. However to maintain this state as the sole ground state one has to continuously tune the chemical potential with increasing system size, otherwise one creates a finite density Fermi surface corresponding to a metal.
    }\label{fig:simplestham}
\end{figure*}

One interesting question is what sort of low-energy excitations exist as a consequence of having the $\ket{W_1}$ state as a ground state. One may show the following statement, where the proof is given in the Supplementary materials~\cite{supp},
\begin{cor}
    For local Hamiltonians with ground state $\ket{W_1}$, the ``boosted $W$-states'' $\ket{W_1^{(m)}}\equiv e^{i\frac{2\pi}{L}m\sum_{i}x_i\hat{n}_i}\ket{W}$ of momentum boosts $2\pi m/L$ ($m$ independent of $L$), will be at most a $O(1/L^2)$ energy expectation value difference above the ground state.
    \label{cor:momentumguys}
\end{cor}
For $U(1)$ and translation symmetric Hamiltonians this statement becomes even stronger as it implies that there are $O(1/L^2)$ low-energy eigenstates. The easiest demonstration of this concept is given by the following free fermion Hamiltonian
\begin{align}
    H&=-\frac{1}{2}\sum_{i}\left[c_i^\dag c_{i+1}+c_{i+1}^\dag c_{i}\right]+\sum_i c_i^\dag c_i,\nonumber\\
    &=\sum_k\left(1-\cos k\right)\,c_k^\dag c_k.
    \label{eq:simpleHam}
\end{align}
The energy dispersion in momentum space is $\varepsilon(k)=1-\cos k$, as shown in Fig.~\ref{fig:simplestham}(b). Here we see that the $\ket{W_1}$ state is a zero mode with the same energy as the empty $\ket{0}$ state. Together they form the smallest ground state manifold possible for the $\ket{W_1}$ state, as necessitated in Theorem~\ref{thm:Wstate}. The low energy excitations of this Hamiltonian are indeed the momentum $k=2\pi m/L$ single particle excitations $c^\dag_k\ket{0}$ which have quadratic dispersion for small $k$, with $k_{\rm min}$ scaling as $1/L$, leading to an $O(1/L^2)$ gap.

This Hamiltonian is special as it represents a critical point, known as a Lifshitz transition, where the Fermi energy precisely touches the bottom of a quadratic dispersion. With these observations in mind we propose a new quantum state signature of a Lifshitz transition
\begin{conj}
    If a translation invariant Hamiltonian is tuned such that a zero-mode $\ket{W_1}$-like state, i.e. a state in the form of $\frac{1}{\sqrt{L}}\sum_i e^{i \frac{2\pi}{L} m x_i}X_i\ket{\alpha}^{\otimes L}$ where $m\in\mathbb{Z}$, and a product state $\ket{\alpha}^{\otimes L}$ become the ground states of the system then the system is at a Lifshitz transition.
\end{conj}
This signature would cover the simple types of Lifshitz transitions such as when a state transitions from an insulator to metal or vice versa, but does not encapsulate other instances such as when a metal or semimetal changes its Fermi surface shape via a Van-Hove singularities or Dirac lines~\cite{voloviklifshitz}. We suspect these more complicated transition may also have an interpretation in terms of the $\ket{W_1}$-like zero mode degeneracy, however the precise statement remains to be formulated. Our conjecture is in alignment with the fact that at these critical points a quadratic (or higher) dispersion necessarily occurs.

\subsection{Gapped Hamiltonians}

Some simple consequences also follow for gapped Hamiltonians for which $\ket{W_1}$ is a ground state. Here we present the most interesting result with more fun facts in the Supplementary materials~\cite{supp}.

A key question is whether the $W$ state can belong to the ground state manifold of a stable gapped phase? Stability of the ground state degeneracy is essential for defining such phases. Generally one desires that the ground state degeneracy is exponentially stable in system size, i.e. a perturbation of magnitude $\lambda\ll1$ creates an exponentially small energy splitting of the degeneracy (such as $O(\lambda^L)$), as is the case for topological orders and fractons~\cite{stabilityHastings,quasiadiabatic,PhysRevB.83.035107}. However the $\ket{W_1}$ degeneracy with $\ket{0}$ is never stable in this sense:
\begin{cor}
The $\ket{W_1}$ state is never a part of a stable gapped ground state manifold, i.e. there exists a perturbation of magnitude $\lambda$ that can lift the ground state degeneracy by an energy proportional to $\lambda$.
\label{cor:degeneracy}
\end{cor}

\begin{proof}
 We have previously shown that $\ket{0}$ must be in the ground state manifold if $\ket{W_1}$ is a ground state. Since this is the case, we may simply add a perturbation $\delta H=-\lambda\sum_i Z_i$ to create an energy gap of $\bra{W_1}H_0+\delta H\ket{W_1}-\bra{0}H_0+\delta H\ket{0}=2\lambda$.
\end{proof}
It follows that these states should be excluded from numerical searches of gapped ground state phases, despite their area law entanglement entropy.

%%%%%%%%%%%%%%%%%%%%%%%%%%%%%%%%%%%%%%%%%%%%%%%%%%%%%%%%%%%%%%%%%%%%%%%%%%
\section{Pathological Parent Hamiltonians}

In this section we demonstrate the limits of Theorem~\ref{thm:Wstate} and its assumptions by constructing converse Hamiltonians that possess the $\ket{W_1}$ state as the unique ground state.

\subsection{Explicit length dependence of parameters}

One key assumption Theorem~\ref{thm:Wstate} makes is that as one takes the thermodynamic limit of the Hamiltonian the addition of new terms does not change the original terms, i.e. no intrinsic $L$-dependence of the individual Hamiltonian parameters. If we break this assumption we can arrive at a Hamiltonian for which the $\ket{W_1}$ is the unique ground state. To do this, we simply modify the critical Lifshitz-transition Hamiltonian in Eq.~\ref{eq:simpleHam} by shifting the chemical potential in $O(1/L^2)$ to create an ever-shrinking Fermi surface, as depicted in Fig.~\ref{fig:simplestham}(c). Such a situation is given by Hamiltonian:
\begin{align}
    H&=-\frac{1}{2}\sum_{i}\left[c_i^\dag c_{i+1}+c_{i+1}^\dag c_{i}\right]+\cos\left(\frac{\pi}{L}\right)\sum_i c_i^\dag c_i\nonumber\\
    &=\sum_k\left(\cos\left(\frac{\pi}{L}\right)-\cos k\right)\,c_k^\dag c_k
    \label{eq:shrinkingFS}
\end{align}
where $\ket{0}$ is an energy $O(1/L^2)$ excitation above $\ket{W_1}$. Here we see that the chemical potential parameter explicitly depends on $L$. We see that this sort of dependence is generally unphysical as it defies the condensed-matter notion of locality since the knowledge of the system size is encoded in non-local operators spanning $O(L)$ sites.

By further breaking the bounded operator norm condition of individual terms one may create an even `sicker' \textit{gapped} Hamiltonian with $\ket{W_1}$ as the unique ground state by explicitly multiplying the terms in Eq.~\ref{eq:shrinkingFS} by $L^2$ such that the $O(1/L^2)$ excitation becomes an $O(1)$ excitation. Hamiltonians with system length dependent parameters, of a similar spirit to Eq.~\ref{eq:shrinkingFS}, for which $\ket{W_1}$ is the unique ground state appear in Refs.~\cite{PhysRevA.99.062104,apel2023simulating,PhysRevA.86.022339,PhysRevA.72.014301}.

\subsection{Locality-breaking}

Naturally if we break locality by allowing terms that couple to sites that are $O(L)$ apart, then we may find Hamiltonians for which $\ket{W_1}$ is the unique, and in fact gapped, ground state. Here we present two notable examples of such a phenomenon.

The first example of this principle is given by a modification of the critical Lifshitz model in Eq.~\ref{eq:simpleHam} with a total charge projector. This procedure results in the non-local gapless Hamiltonian
\begin{align}
    H=\lambda\prod_i \left(1-2n_i\right)-\frac{1}{2}\sum_{i}\left[c_i^\dag c_{i+1}+c_{i+1}^\dag c_{i}\right]+\sum_i c_i^\dag c_i,
\end{align}
where $\lambda>0$ which projects states in the odd charge sector (such as $\ket{W_1}$) to a lower energy state by $2\lambda$ as compared to even charge sectors (such as $\ket{0}$). This lifts the degeneracy of the ground state while maintaining the gapless spectrum at the cost of the non-local charge projector.

The second example is an all-to-all 2-body Hamiltonian with unbounded norm, as presented in Ref.~\cite{frowis2013certifiability},
\begin{align}
    H&=\left(1-\sum_i \frac{1}{2}(1-Z_i)\right)^2-J^2\quad,
    \label{eq:nonlocalH}
\end{align}
where $J^2=(\sum_i X_i)^2+(\sum_i Y_i)^2+(\sum_i Z_i)^2$ for which the Dicke states~\cite{Dicke} $\ket{L,m}$ are the eigenstates for $J^2$ and $J_z$ with eigenvalues $L/2 (L/2+1)$ and $L/2-m$, respectively. Here $\ket{W_1}=\ket{L,1}$ is the gapped lowest energy ground state of Eq.~\ref{eq:nonlocalH} since the first term favours states with total charge one, and the second term lifts the degeneracy between $\ket{W_1}$ and its momentum boosted states. The unbounded norm of operators renders it difficult to define a gap in the usual condensed-matter sense and is similarly pathological to the case in Eq.~\ref{eq:shrinkingFS} when the terms are multiplied by $L^2$.

Recently, a more mild notion of non-local Hamiltonian has received attention, due to its connection with low-density parity-check (LDPC) codes \cite{dinur2022good,Anshu_2023,leverrier2022quantum}. We can define a low-density Hamiltonian to be one whose terms each involve a bounded number of sites, and for which each site participates in a bounded number of terms. The latter condition rules out the Hamiltonian \eqref{eq:nonlocalH}. We could also soften this condition by requiring only that the sum of the norms is bounded in each case. We do not know at this time whether the $W$ state is a unique ground state of such a Hamiltonian (or whether the $W$ state is an LDPC code).

Low-density Hamiltonians can prepare more kinds of states than local Hamiltonian systems. For example, our method of argument also shows that $\frac{1}{\sqrt{2}} (X_1+X_{L/2}) \ket{0}$ is not the unique ground state of any local Hamiltonian system. However, it is the unique (even gapped!) ground state of the low-density Hamiltonian $H = 2Z_1 Z_{L/2} -\sum_i Z_i$.

%%%%%%%%%%%%%%%%%%%%%%%%%%%%%%%%%%%%%%%%%%%%%%%%%%%%%%%%%%%%%%%%%%%%%%%%%%
\section{Discussion}

The relationship between $\ket{0}$ and the $W$ state can be generalized. Given a state $\ket{\psi}$, we can consider $\ket{\psi'} = \sum_i \cO_i \ket{\psi}$. We propose to call these \emph{wh\textoverline{a}nau-states}\footnote{Wh\textoverline{a}nau, pronounced [\textprimstress fa\textlengthmark na\textbaru], is a Maori word for family.}. In higher dimensions, we can also consider applying operators on subspaces, such as $\frac{1}{\sqrt{L_x}}\sum_{i=1}^{L_x} \prod_{j = 1}^{L_y} X_j \ket{0}$ in a 2d square lattice. This can be extended to generate a whole ``family tree'' of $\ket{\psi}$, and we expect analogous ``groundstateability'' results along this whole tree, such that a state can only be a ground state if its ancestors are as well. We sketch an argument for this in the supplemental material.

An interesting property concerning $\ket{0}$ and $\ket{W}$ is that for all regions $R$, the reduced density matrices $\rho_R^0$ and $\rho_R^W$ converge in the $L \to \infty$ limit. \footnote{This convergence is not uniform in $R$: if we set $|R| = L/2$, the entanglement entropy of $\ket{W}$ is $\log 2$.}. In particular, the expectation value of any fixed operator in $\ket{0}$ and $\ket{W}$ will converge as $L \to \infty$. 

Although from this point of view, the $W$ state is indistinguishable from a product state in the thermodynamic limit, it is still long-range entangled (this can be proven by considering the boosted $W$ states, which have nonzero momentum, and then applying the results of \cite{GioiaWang}). This presents a challenge for the formal understanding of states in the thermodynamic limit as functions on the algebra of local observables \cite{bratteli2012operator,ruelle2001topics}, since these states are identified, while being physically distinct. What is needed, it seems, is a theory of the thermodynamic limit which can also keep track of finite-size corrections.

\begin{acknowledgments}
The authors especially thank Omar Abdelghani, Xie Chen, Nick G. Jones, Ruben Verresen, and  Chong Wang for very useful insight and suggestions. LG is grateful to Leonardo A. Lessa, Sanjay Moudgalya, and Sergey Syzranov for related discussions. We are grateful to the Perimeter Institute for Theoretical Physics and Institut des Hautes \'Etudes Scientifiques for hosting us during part of this work. LG also thanks the graduate fellowship program at the Kavli Institute for Theoretical Physics.
\end{acknowledgments}

\bibliography{references}

\begin{thebibliography}{29}%
\makeatletter
\providecommand \@ifxundefined [1]{%
 \@ifx{#1\undefined}
}%
\providecommand \@ifnum [1]{%
 \ifnum #1\expandafter \@firstoftwo
 \else \expandafter \@secondoftwo
 \fi
}%
\providecommand \@ifx [1]{%
 \ifx #1\expandafter \@firstoftwo
 \else \expandafter \@secondoftwo
 \fi
}%
\providecommand \natexlab [1]{#1}%
\providecommand \enquote  [1]{``#1''}%
\providecommand \bibnamefont  [1]{#1}%
\providecommand \bibfnamefont [1]{#1}%
\providecommand \citenamefont [1]{#1}%
\providecommand \href@noop [0]{\@secondoftwo}%
\providecommand \href [0]{\begingroup \@sanitize@url \@href}%
\providecommand \@href[1]{\@@startlink{#1}\@@href}%
\providecommand \@@href[1]{\endgroup#1\@@endlink}%
\providecommand \@sanitize@url [0]{\catcode `\\12\catcode `\$12\catcode
  `\&12\catcode `\#12\catcode `\^12\catcode `\_12\catcode `\%12\relax}%
\providecommand \@@startlink[1]{}%
\providecommand \@@endlink[0]{}%
\providecommand \url  [0]{\begingroup\@sanitize@url \@url }%
\providecommand \@url [1]{\endgroup\@href {#1}{\urlprefix }}%
\providecommand \urlprefix  [0]{URL }%
\providecommand \Eprint [0]{\href }%
\providecommand \doibase [0]{https://doi.org/}%
\providecommand \selectlanguage [0]{\@gobble}%
\providecommand \bibinfo  [0]{\@secondoftwo}%
\providecommand \bibfield  [0]{\@secondoftwo}%
\providecommand \translation [1]{[#1]}%
\providecommand \BibitemOpen [0]{}%
\providecommand \bibitemStop [0]{}%
\providecommand \bibitemNoStop [0]{.\EOS\space}%
\providecommand \EOS [0]{\spacefactor3000\relax}%
\providecommand \BibitemShut  [1]{\csname bibitem#1\endcsname}%
\let\auto@bib@innerbib\@empty
%</preamble>
\bibitem [{\citenamefont {Joo}\ \emph {et~al.}(2003)\citenamefont {Joo},
  \citenamefont {Park}, \citenamefont {Oh},\ and\ \citenamefont
  {Kim}}]{joo2003quantum}%
  \BibitemOpen
  \bibfield  {author} {\bibinfo {author} {\bibfnamefont {J.}~\bibnamefont
  {Joo}}, \bibinfo {author} {\bibfnamefont {Y.-J.}\ \bibnamefont {Park}},
  \bibinfo {author} {\bibfnamefont {S.}~\bibnamefont {Oh}},\ and\ \bibinfo
  {author} {\bibfnamefont {J.}~\bibnamefont {Kim}},\ }\bibfield  {title}
  {\bibinfo {title} {Quantum teleportation via a w state},\ }\href@noop {}
  {\bibfield  {journal} {\bibinfo  {journal} {New Journal of Physics}\ }\textbf
  {\bibinfo {volume} {5}},\ \bibinfo {pages} {136} (\bibinfo {year}
  {2003})}\BibitemShut {NoStop}%
\bibitem [{\citenamefont {Fr{\"o}wis}\ \emph {et~al.}(2013)\citenamefont
  {Fr{\"o}wis}, \citenamefont {Van~den Nest},\ and\ \citenamefont
  {D{\"u}r}}]{frowis2013certifiability}%
  \BibitemOpen
  \bibfield  {author} {\bibinfo {author} {\bibfnamefont {F.}~\bibnamefont
  {Fr{\"o}wis}}, \bibinfo {author} {\bibfnamefont {M.}~\bibnamefont {Van~den
  Nest}},\ and\ \bibinfo {author} {\bibfnamefont {W.}~\bibnamefont {D{\"u}r}},\
  }\bibfield  {title} {\bibinfo {title} {Certifiability criterion for
  large-scale quantum systems},\ }\href@noop {} {\bibfield  {journal} {\bibinfo
   {journal} {New Journal of Physics}\ }\textbf {\bibinfo {volume} {15}},\
  \bibinfo {pages} {113011} (\bibinfo {year} {2013})}\BibitemShut {NoStop}%
\bibitem [{\citenamefont {Faist}\ \emph {et~al.}(2020)\citenamefont {Faist},
  \citenamefont {Nezami}, \citenamefont {Albert}, \citenamefont {Salton},
  \citenamefont {Pastawski}, \citenamefont {Hayden},\ and\ \citenamefont
  {Preskill}}]{Faist_2020}%
  \BibitemOpen
  \bibfield  {author} {\bibinfo {author} {\bibfnamefont {P.}~\bibnamefont
  {Faist}}, \bibinfo {author} {\bibfnamefont {S.}~\bibnamefont {Nezami}},
  \bibinfo {author} {\bibfnamefont {V.~V.}\ \bibnamefont {Albert}}, \bibinfo
  {author} {\bibfnamefont {G.}~\bibnamefont {Salton}}, \bibinfo {author}
  {\bibfnamefont {F.}~\bibnamefont {Pastawski}}, \bibinfo {author}
  {\bibfnamefont {P.}~\bibnamefont {Hayden}},\ and\ \bibinfo {author}
  {\bibfnamefont {J.}~\bibnamefont {Preskill}},\ }\bibfield  {title} {\bibinfo
  {title} {Continuous symmetries and approximate quantum error correction},\
  }\bibfield  {journal} {\bibinfo  {journal} {Physical Review X}\ }\textbf
  {\bibinfo {volume} {10}},\ \href {https://doi.org/10.1103/physrevx.10.041018}
  {10.1103/physrevx.10.041018} (\bibinfo {year} {2020})\BibitemShut {NoStop}%
\bibitem [{\citenamefont {Zang}\ \emph
  {et~al.}(2015{\natexlab{a}})\citenamefont {Zang}, \citenamefont {Yang},
  \citenamefont {Ozaydin}, \citenamefont {Song},\ and\ \citenamefont
  {Cao}}]{lightmatterW}%
  \BibitemOpen
  \bibfield  {author} {\bibinfo {author} {\bibfnamefont {X.-P.}\ \bibnamefont
  {Zang}}, \bibinfo {author} {\bibfnamefont {M.}~\bibnamefont {Yang}}, \bibinfo
  {author} {\bibfnamefont {F.}~\bibnamefont {Ozaydin}}, \bibinfo {author}
  {\bibfnamefont {W.}~\bibnamefont {Song}},\ and\ \bibinfo {author}
  {\bibfnamefont {Z.-L.}\ \bibnamefont {Cao}},\ }\bibfield  {title} {\bibinfo
  {title} {Generating multi-atom entangled w states via light-matter interface
  based fusion mechanism},\ }\href {https://doi.org/10.1038/srep16245}
  {\bibfield  {journal} {\bibinfo  {journal} {Scientific Reports}\ }\textbf
  {\bibinfo {volume} {5}},\ \bibinfo {pages} {16245} (\bibinfo {year}
  {2015}{\natexlab{a}})}\BibitemShut {NoStop}%
\bibitem [{\citenamefont {Cirac}\ \emph {et~al.}(2021)\citenamefont {Cirac},
  \citenamefont {P\'erez-Garc\'{\i}a}, \citenamefont {Schuch},\ and\
  \citenamefont {Verstraete}}]{RevModPhys.93.045003}%
  \BibitemOpen
  \bibfield  {author} {\bibinfo {author} {\bibfnamefont {J.~I.}\ \bibnamefont
  {Cirac}}, \bibinfo {author} {\bibfnamefont {D.}~\bibnamefont
  {P\'erez-Garc\'{\i}a}}, \bibinfo {author} {\bibfnamefont {N.}~\bibnamefont
  {Schuch}},\ and\ \bibinfo {author} {\bibfnamefont {F.}~\bibnamefont
  {Verstraete}},\ }\bibfield  {title} {\bibinfo {title} {Matrix product states
  and projected entangled pair states: Concepts, symmetries, theorems},\ }\href
  {https://doi.org/10.1103/RevModPhys.93.045003} {\bibfield  {journal}
  {\bibinfo  {journal} {Rev. Mod. Phys.}\ }\textbf {\bibinfo {volume} {93}},\
  \bibinfo {pages} {045003} (\bibinfo {year} {2021})}\BibitemShut {NoStop}%
\bibitem [{\citenamefont {Odavić}\ \emph {et~al.}(2023)\citenamefont
  {Odavić}, \citenamefont {Haug}, \citenamefont {Torre}, \citenamefont
  {Hamma}, \citenamefont {Franchini},\ and\ \citenamefont
  {Giampaolo}}]{10.21468/SciPostPhys.15.4.131}%
  \BibitemOpen
  \bibfield  {author} {\bibinfo {author} {\bibfnamefont {J.}~\bibnamefont
  {Odavić}}, \bibinfo {author} {\bibfnamefont {T.}~\bibnamefont {Haug}},
  \bibinfo {author} {\bibfnamefont {G.}~\bibnamefont {Torre}}, \bibinfo
  {author} {\bibfnamefont {A.}~\bibnamefont {Hamma}}, \bibinfo {author}
  {\bibfnamefont {F.}~\bibnamefont {Franchini}},\ and\ \bibinfo {author}
  {\bibfnamefont {S.~M.}\ \bibnamefont {Giampaolo}},\ }\bibfield  {title}
  {\bibinfo {title} {{Complexity of frustration: A new source of non-local
  non-stabilizerness}},\ }\href {https://doi.org/10.21468/SciPostPhys.15.4.131}
  {\bibfield  {journal} {\bibinfo  {journal} {SciPost Phys.}\ }\textbf
  {\bibinfo {volume} {15}},\ \bibinfo {pages} {131} (\bibinfo {year}
  {2023})}\BibitemShut {NoStop}%
\bibitem [{\citenamefont {Dür}\ \emph {et~al.}(2000)\citenamefont {Dür},
  \citenamefont {Vidal},\ and\ \citenamefont {Cirac}}]{D_r_2000}%
  \BibitemOpen
  \bibfield  {author} {\bibinfo {author} {\bibfnamefont {W.}~\bibnamefont
  {Dür}}, \bibinfo {author} {\bibfnamefont {G.}~\bibnamefont {Vidal}},\ and\
  \bibinfo {author} {\bibfnamefont {J.~I.}\ \bibnamefont {Cirac}},\ }\bibfield
  {title} {\bibinfo {title} {Three qubits can be entangled in two inequivalent
  ways},\ }\bibfield  {journal} {\bibinfo  {journal} {Physical Review A}\
  }\textbf {\bibinfo {volume} {62}},\ \href
  {https://doi.org/10.1103/physreva.62.062314} {10.1103/physreva.62.062314}
  (\bibinfo {year} {2000})\BibitemShut {NoStop}%
\bibitem [{\citenamefont {Zang}\ \emph
  {et~al.}(2015{\natexlab{b}})\citenamefont {Zang}, \citenamefont {Yang},
  \citenamefont {Ozaydin}, \citenamefont {Song},\ and\ \citenamefont
  {Cao}}]{zang2015generating}%
  \BibitemOpen
  \bibfield  {author} {\bibinfo {author} {\bibfnamefont {X.-P.}\ \bibnamefont
  {Zang}}, \bibinfo {author} {\bibfnamefont {M.}~\bibnamefont {Yang}}, \bibinfo
  {author} {\bibfnamefont {F.}~\bibnamefont {Ozaydin}}, \bibinfo {author}
  {\bibfnamefont {W.}~\bibnamefont {Song}},\ and\ \bibinfo {author}
  {\bibfnamefont {Z.-L.}\ \bibnamefont {Cao}},\ }\bibfield  {title} {\bibinfo
  {title} {Generating multi-atom entangled w states via light-matter interface
  based fusion mechanism},\ }\href@noop {} {\bibfield  {journal} {\bibinfo
  {journal} {Scientific Reports}\ }\textbf {\bibinfo {volume} {5}},\ \bibinfo
  {pages} {16245} (\bibinfo {year} {2015}{\natexlab{b}})}\BibitemShut {NoStop}%
\bibitem [{\citenamefont {Cruz}\ \emph {et~al.}(2019)\citenamefont {Cruz},
  \citenamefont {Fournier}, \citenamefont {Gremion}, \citenamefont {Jeannerot},
  \citenamefont {Komagata}, \citenamefont {Tosic}, \citenamefont
  {Thiesbrummel}, \citenamefont {Chan}, \citenamefont {Macris}, \citenamefont
  {Dupertuis},\ and\ \citenamefont {Javerzac-Galy}}]{Cruz_2019}%
  \BibitemOpen
  \bibfield  {author} {\bibinfo {author} {\bibfnamefont {D.}~\bibnamefont
  {Cruz}}, \bibinfo {author} {\bibfnamefont {R.}~\bibnamefont {Fournier}},
  \bibinfo {author} {\bibfnamefont {F.}~\bibnamefont {Gremion}}, \bibinfo
  {author} {\bibfnamefont {A.}~\bibnamefont {Jeannerot}}, \bibinfo {author}
  {\bibfnamefont {K.}~\bibnamefont {Komagata}}, \bibinfo {author}
  {\bibfnamefont {T.}~\bibnamefont {Tosic}}, \bibinfo {author} {\bibfnamefont
  {J.}~\bibnamefont {Thiesbrummel}}, \bibinfo {author} {\bibfnamefont {C.~L.}\
  \bibnamefont {Chan}}, \bibinfo {author} {\bibfnamefont {N.}~\bibnamefont
  {Macris}}, \bibinfo {author} {\bibfnamefont {M.-A.}\ \bibnamefont
  {Dupertuis}},\ and\ \bibinfo {author} {\bibfnamefont {C.}~\bibnamefont
  {Javerzac-Galy}},\ }\bibfield  {title} {\bibinfo {title} {Efficient quantum
  algorithms for {GHZ} and w states, and implementation on the {IBM} quantum
  computer},\ }\bibfield  {journal} {\bibinfo  {journal} {Advanced Quantum
  Technologies}\ }\textbf {\bibinfo {volume} {2}},\ \href
  {https://doi.org/10.1002/qute.201900015} {10.1002/qute.201900015} (\bibinfo
  {year} {2019})\BibitemShut {NoStop}%
\bibitem [{\citenamefont {Heo}\ \emph {et~al.}(2019)\citenamefont {Heo},
  \citenamefont {Hong}, \citenamefont {Choi},\ and\ \citenamefont
  {Hong}}]{heo2019scheme}%
  \BibitemOpen
  \bibfield  {author} {\bibinfo {author} {\bibfnamefont {J.}~\bibnamefont
  {Heo}}, \bibinfo {author} {\bibfnamefont {C.}~\bibnamefont {Hong}}, \bibinfo
  {author} {\bibfnamefont {S.-G.}\ \bibnamefont {Choi}},\ and\ \bibinfo
  {author} {\bibfnamefont {J.-P.}\ \bibnamefont {Hong}},\ }\bibfield  {title}
  {\bibinfo {title} {Scheme for generation of three-photon entangled w state
  assisted by cross-kerr nonlinearity and quantum dot},\ }\href@noop {}
  {\bibfield  {journal} {\bibinfo  {journal} {Scientific reports}\ }\textbf
  {\bibinfo {volume} {9}},\ \bibinfo {pages} {10151} (\bibinfo {year}
  {2019})}\BibitemShut {NoStop}%
\bibitem [{\citenamefont {Perez-Garcia}\ \emph {et~al.}(2007)\citenamefont
  {Perez-Garcia}, \citenamefont {Verstraete}, \citenamefont {Wolf},\ and\
  \citenamefont {Cirac}}]{mpspaper}%
  \BibitemOpen
  \bibfield  {author} {\bibinfo {author} {\bibfnamefont {D.}~\bibnamefont
  {Perez-Garcia}}, \bibinfo {author} {\bibfnamefont {F.}~\bibnamefont
  {Verstraete}}, \bibinfo {author} {\bibfnamefont {M.~M.}\ \bibnamefont
  {Wolf}},\ and\ \bibinfo {author} {\bibfnamefont {J.~I.}\ \bibnamefont
  {Cirac}},\ }\bibfield  {title} {\bibinfo {title} {Matrix product state
  representations},\ }\href@noop {} {\bibfield  {journal} {\bibinfo  {journal}
  {Quantum Info. Comput.}\ }\textbf {\bibinfo {volume} {7}},\ \bibinfo {pages}
  {401–430} (\bibinfo {year} {2007})}\BibitemShut {NoStop}%
\bibitem [{\citenamefont {Volovik}(2017)}]{voloviklifshitz}%
  \BibitemOpen
  \bibfield  {author} {\bibinfo {author} {\bibfnamefont {G.~E.}\ \bibnamefont
  {Volovik}},\ }\bibfield  {title} {\bibinfo {title} {{Topological Lifshitz
  transitions}},\ }\href {https://doi.org/10.1063/1.4974185} {\bibfield
  {journal} {\bibinfo  {journal} {Low Temperature Physics}\ }\textbf {\bibinfo
  {volume} {43}},\ \bibinfo {pages} {47} (\bibinfo {year} {2017})},\ \Eprint
  {https://arxiv.org/abs/https://pubs.aip.org/aip/ltp/article-pdf/43/1/47/15722194/47\_1\_online.pdf}
  {https://pubs.aip.org/aip/ltp/article-pdf/43/1/47/15722194/47\_1\_online.pdf}
  \BibitemShut {NoStop}%
\bibitem [{\citenamefont {Daki\ifmmode~\acute{c}\else \'{c}\fi{}}\ and\
  \citenamefont {Radonji\ifmmode~\acute{c}\else
  \'{c}\fi{}}(2017)}]{macroscopicsuperpositions}%
  \BibitemOpen
  \bibfield  {author} {\bibinfo {author} {\bibfnamefont {B.}~\bibnamefont
  {Daki\ifmmode~\acute{c}\else \'{c}\fi{}}}\ and\ \bibinfo {author}
  {\bibfnamefont {M.}~\bibnamefont {Radonji\ifmmode~\acute{c}\else
  \'{c}\fi{}}},\ }\bibfield  {title} {\bibinfo {title} {Macroscopic
  superpositions as quantum ground states},\ }\href
  {https://doi.org/10.1103/PhysRevLett.119.090401} {\bibfield  {journal}
  {\bibinfo  {journal} {Phys. Rev. Lett.}\ }\textbf {\bibinfo {volume} {119}},\
  \bibinfo {pages} {090401} (\bibinfo {year} {2017})}\BibitemShut {NoStop}%
\bibitem [{sup()}]{supp}%
  \BibitemOpen
  \href@noop {} {}\bibinfo {note} {See Supplemental Material for proof of
  Proposition 1, the sketch of the generalisation for the main proof, proof of
  Corollary 1, and various fun facts alongside interesting example
  Hamiltonians.}\BibitemShut {Stop}%
\bibitem [{\citenamefont {Bravyi}\ and\ \citenamefont
  {Hastings}(2011)}]{stabilityHastings}%
  \BibitemOpen
  \bibfield  {author} {\bibinfo {author} {\bibfnamefont {S.}~\bibnamefont
  {Bravyi}}\ and\ \bibinfo {author} {\bibfnamefont {M.~B.}\ \bibnamefont
  {Hastings}},\ }\bibfield  {title} {\bibinfo {title} {A short proof of
  stability of topological order under local perturbations},\ }\href
  {https://doi.org/10.1007/s00220-011-1346-2} {\bibfield  {journal} {\bibinfo
  {journal} {Communications in Mathematical Physics}\ }\textbf {\bibinfo
  {volume} {307}},\ \bibinfo {pages} {609} (\bibinfo {year}
  {2011})}\BibitemShut {NoStop}%
\bibitem [{\citenamefont {Hastings}\ and\ \citenamefont
  {Wen}(2005)}]{quasiadiabatic}%
  \BibitemOpen
  \bibfield  {author} {\bibinfo {author} {\bibfnamefont {M.~B.}\ \bibnamefont
  {Hastings}}\ and\ \bibinfo {author} {\bibfnamefont {X.-G.}\ \bibnamefont
  {Wen}},\ }\bibfield  {title} {\bibinfo {title} {Quasiadiabatic continuation
  of quantum states: The stability of topological ground-state degeneracy and
  emergent gauge invariance},\ }\href
  {https://doi.org/10.1103/PhysRevB.72.045141} {\bibfield  {journal} {\bibinfo
  {journal} {Phys. Rev. B}\ }\textbf {\bibinfo {volume} {72}},\ \bibinfo
  {pages} {045141} (\bibinfo {year} {2005})}\BibitemShut {NoStop}%
\bibitem [{\citenamefont {Chen}\ \emph {et~al.}(2011)\citenamefont {Chen},
  \citenamefont {Gu},\ and\ \citenamefont {Wen}}]{PhysRevB.83.035107}%
  \BibitemOpen
  \bibfield  {author} {\bibinfo {author} {\bibfnamefont {X.}~\bibnamefont
  {Chen}}, \bibinfo {author} {\bibfnamefont {Z.-C.}\ \bibnamefont {Gu}},\ and\
  \bibinfo {author} {\bibfnamefont {X.-G.}\ \bibnamefont {Wen}},\ }\bibfield
  {title} {\bibinfo {title} {Classification of gapped symmetric phases in
  one-dimensional spin systems},\ }\href
  {https://doi.org/10.1103/PhysRevB.83.035107} {\bibfield  {journal} {\bibinfo
  {journal} {Phys. Rev. B}\ }\textbf {\bibinfo {volume} {83}},\ \bibinfo
  {pages} {035107} (\bibinfo {year} {2011})}\BibitemShut {NoStop}%
\bibitem [{\citenamefont {Karuvade}\ \emph {et~al.}(2019)\citenamefont
  {Karuvade}, \citenamefont {Johnson}, \citenamefont {Ticozzi},\ and\
  \citenamefont {Viola}}]{PhysRevA.99.062104}%
  \BibitemOpen
  \bibfield  {author} {\bibinfo {author} {\bibfnamefont {S.}~\bibnamefont
  {Karuvade}}, \bibinfo {author} {\bibfnamefont {P.~D.}\ \bibnamefont
  {Johnson}}, \bibinfo {author} {\bibfnamefont {F.}~\bibnamefont {Ticozzi}},\
  and\ \bibinfo {author} {\bibfnamefont {L.}~\bibnamefont {Viola}},\ }\bibfield
   {title} {\bibinfo {title} {Uniquely determined pure quantum states need not
  be unique ground states of quasi-local hamiltonians},\ }\href
  {https://doi.org/10.1103/PhysRevA.99.062104} {\bibfield  {journal} {\bibinfo
  {journal} {Phys. Rev. A}\ }\textbf {\bibinfo {volume} {99}},\ \bibinfo
  {pages} {062104} (\bibinfo {year} {2019})}\BibitemShut {NoStop}%
\bibitem [{\citenamefont {Apel}\ and\ \citenamefont
  {Baspin}(2023)}]{apel2023simulating}%
  \BibitemOpen
  \bibfield  {author} {\bibinfo {author} {\bibfnamefont {H.}~\bibnamefont
  {Apel}}\ and\ \bibinfo {author} {\bibfnamefont {N.}~\bibnamefont {Baspin}},\
  }\href@noop {} {\bibinfo {title} {Simulating ldpc code hamiltonians on 2d
  lattices}} (\bibinfo {year} {2023}),\ \Eprint
  {https://arxiv.org/abs/2308.13277} {arXiv:2308.13277 [quant-ph]} \BibitemShut
  {NoStop}%
\bibitem [{\citenamefont {Chen}\ \emph {et~al.}(2012)\citenamefont {Chen},
  \citenamefont {Ji}, \citenamefont {Zeng},\ and\ \citenamefont
  {Zhou}}]{PhysRevA.86.022339}%
  \BibitemOpen
  \bibfield  {author} {\bibinfo {author} {\bibfnamefont {J.}~\bibnamefont
  {Chen}}, \bibinfo {author} {\bibfnamefont {Z.}~\bibnamefont {Ji}}, \bibinfo
  {author} {\bibfnamefont {B.}~\bibnamefont {Zeng}},\ and\ \bibinfo {author}
  {\bibfnamefont {D.~L.}\ \bibnamefont {Zhou}},\ }\bibfield  {title} {\bibinfo
  {title} {From ground states to local hamiltonians},\ }\href
  {https://doi.org/10.1103/PhysRevA.86.022339} {\bibfield  {journal} {\bibinfo
  {journal} {Phys. Rev. A}\ }\textbf {\bibinfo {volume} {86}},\ \bibinfo
  {pages} {022339} (\bibinfo {year} {2012})}\BibitemShut {NoStop}%
\bibitem [{\citenamefont {Bru\ss{}}\ \emph {et~al.}(2005)\citenamefont
  {Bru\ss{}}, \citenamefont {Datta}, \citenamefont {Ekert}, \citenamefont
  {Kwek},\ and\ \citenamefont {Macchiavello}}]{PhysRevA.72.014301}%
  \BibitemOpen
  \bibfield  {author} {\bibinfo {author} {\bibfnamefont {D.}~\bibnamefont
  {Bru\ss{}}}, \bibinfo {author} {\bibfnamefont {N.}~\bibnamefont {Datta}},
  \bibinfo {author} {\bibfnamefont {A.}~\bibnamefont {Ekert}}, \bibinfo
  {author} {\bibfnamefont {L.~C.}\ \bibnamefont {Kwek}},\ and\ \bibinfo
  {author} {\bibfnamefont {C.}~\bibnamefont {Macchiavello}},\ }\bibfield
  {title} {\bibinfo {title} {Multipartite entanglement in quantum spin
  chains},\ }\href {https://doi.org/10.1103/PhysRevA.72.014301} {\bibfield
  {journal} {\bibinfo  {journal} {Phys. Rev. A}\ }\textbf {\bibinfo {volume}
  {72}},\ \bibinfo {pages} {014301} (\bibinfo {year} {2005})}\BibitemShut
  {NoStop}%
\bibitem [{\citenamefont {Dicke}(1954)}]{Dicke}%
  \BibitemOpen
  \bibfield  {author} {\bibinfo {author} {\bibfnamefont {R.~H.}\ \bibnamefont
  {Dicke}},\ }\bibfield  {title} {\bibinfo {title} {Coherence in spontaneous
  radiation processes},\ }\href {https://doi.org/10.1103/PhysRev.93.99}
  {\bibfield  {journal} {\bibinfo  {journal} {Phys. Rev.}\ }\textbf {\bibinfo
  {volume} {93}},\ \bibinfo {pages} {99} (\bibinfo {year} {1954})}\BibitemShut
  {NoStop}%
\bibitem [{\citenamefont {Dinur}\ \emph {et~al.}(2022)\citenamefont {Dinur},
  \citenamefont {Hsieh}, \citenamefont {Lin},\ and\ \citenamefont
  {Vidick}}]{dinur2022good}%
  \BibitemOpen
  \bibfield  {author} {\bibinfo {author} {\bibfnamefont {I.}~\bibnamefont
  {Dinur}}, \bibinfo {author} {\bibfnamefont {M.-H.}\ \bibnamefont {Hsieh}},
  \bibinfo {author} {\bibfnamefont {T.-C.}\ \bibnamefont {Lin}},\ and\ \bibinfo
  {author} {\bibfnamefont {T.}~\bibnamefont {Vidick}},\ }\href@noop {}
  {\bibinfo {title} {Good quantum ldpc codes with linear time decoders}}
  (\bibinfo {year} {2022}),\ \Eprint {https://arxiv.org/abs/2206.07750}
  {arXiv:2206.07750 [quant-ph]} \BibitemShut {NoStop}%
\bibitem [{\citenamefont {Anshu}\ \emph {et~al.}(2023)\citenamefont {Anshu},
  \citenamefont {Breuckmann},\ and\ \citenamefont {Nirkhe}}]{Anshu_2023}%
  \BibitemOpen
  \bibfield  {author} {\bibinfo {author} {\bibfnamefont {A.}~\bibnamefont
  {Anshu}}, \bibinfo {author} {\bibfnamefont {N.~P.}\ \bibnamefont
  {Breuckmann}},\ and\ \bibinfo {author} {\bibfnamefont {C.}~\bibnamefont
  {Nirkhe}},\ }\bibfield  {title} {\bibinfo {title} {{NLTS} hamiltonians from
  good quantum codes},\ }in\ \href {https://doi.org/10.1145/3564246.3585114}
  {\emph {\bibinfo {booktitle} {Proceedings of the 55th Annual {ACM} Symposium
  on Theory of Computing}}}\ (\bibinfo  {publisher} {{ACM}},\ \bibinfo {year}
  {2023})\BibitemShut {NoStop}%
\bibitem [{\citenamefont {Leverrier}\ and\ \citenamefont
  {Zémor}(2022)}]{leverrier2022quantum}%
  \BibitemOpen
  \bibfield  {author} {\bibinfo {author} {\bibfnamefont {A.}~\bibnamefont
  {Leverrier}}\ and\ \bibinfo {author} {\bibfnamefont {G.}~\bibnamefont
  {Zémor}},\ }\href@noop {} {\bibinfo {title} {Quantum tanner codes}}
  (\bibinfo {year} {2022}),\ \Eprint {https://arxiv.org/abs/2202.13641}
  {arXiv:2202.13641 [quant-ph]} \BibitemShut {NoStop}%
\bibitem [{\citenamefont {Gioia}\ and\ \citenamefont {Wang}(2022)}]{GioiaWang}%
  \BibitemOpen
  \bibfield  {author} {\bibinfo {author} {\bibfnamefont {L.}~\bibnamefont
  {Gioia}}\ and\ \bibinfo {author} {\bibfnamefont {C.}~\bibnamefont {Wang}},\
  }\bibfield  {title} {\bibinfo {title} {Nonzero momentum requires long-range
  entanglement},\ }\href {https://doi.org/10.1103/PhysRevX.12.031007}
  {\bibfield  {journal} {\bibinfo  {journal} {Phys. Rev. X}\ }\textbf {\bibinfo
  {volume} {12}},\ \bibinfo {pages} {031007} (\bibinfo {year}
  {2022})}\BibitemShut {NoStop}%
\bibitem [{\citenamefont {Bratteli}\ and\ \citenamefont
  {Robinson}(2012)}]{bratteli2012operator}%
  \BibitemOpen
  \bibfield  {author} {\bibinfo {author} {\bibfnamefont {O.}~\bibnamefont
  {Bratteli}}\ and\ \bibinfo {author} {\bibfnamefont {D.~W.}\ \bibnamefont
  {Robinson}},\ }\href@noop {} {\emph {\bibinfo {title} {Operator algebras and
  quantum statistical mechanics: Volume 1: C*-and W*-Algebras. Symmetry Groups.
  Decomposition of States}}}\ (\bibinfo  {publisher} {Springer Science \&
  Business Media},\ \bibinfo {year} {2012})\BibitemShut {NoStop}%
\bibitem [{\citenamefont {Ruelle}(2001)}]{ruelle2001topics}%
  \BibitemOpen
  \bibfield  {author} {\bibinfo {author} {\bibfnamefont {D.}~\bibnamefont
  {Ruelle}},\ }\href@noop {} {\bibinfo {title} {Topics in quantum statistical
  mechanics and operator algebras}} (\bibinfo {year} {2001}),\ \Eprint
  {https://arxiv.org/abs/math-ph/0107009} {arXiv:math-ph/0107009 [math-ph]}
  \BibitemShut {NoStop}%
\bibitem [{\citenamefont {Watanabe}(2019)}]{harukibloch}%
  \BibitemOpen
  \bibfield  {author} {\bibinfo {author} {\bibfnamefont {H.}~\bibnamefont
  {Watanabe}},\ }\bibfield  {title} {\bibinfo {title} {A proof of the bloch
  theorem for lattice models},\ }\href
  {https://doi.org/10.1007/s10955-019-02386-1} {\bibfield  {journal} {\bibinfo
  {journal} {Journal of Statistical Physics}\ }\textbf {\bibinfo {volume}
  {177}},\ \bibinfo {pages} {717} (\bibinfo {year} {2019})}\BibitemShut
  {NoStop}%
\end{thebibliography}%

\onecolumngrid
\appendix

\section{Proof of Proposition 1}

For convenience, we normalize $H$ so that $\langle 0 | h_k |0 \rangle = 0$ for all $k$ (this can be achieved by normal ordering $h_k$ and then subtracting off the constant piece). Let us also adopt the shorthand $\langle \cO \rangle$ for expectation values in the state $\ket{0}$.

Our quantity of interest is
\[\Delta_n = -{L \choose n}^{-1} \sum_{i_1 < \cdots < i_n} \sum_{j_1 < \cdots < j_n} \sum_k \langle X_{i_1} \cdots X_{i_n} h_k X_{j_1} \cdots X_{j_n}\rangle.\]
Since $\ket{0}$ is a product state, where these operators are not overlapping we can split the expectation value into a product of expectation values. In particular, $\langle X_i \rangle = 0$, so we only get nonzero contributions from where each $X$ either collides with another $X$, giving $X^2 = 1$, or with $h_k$ (compare \eqref{eqnXvanishing}). The sum over $j$'s can then be replaced with a sum with $|i_m - j_m| \le 2l+1$, where recall $l$ is the range of the Hamiltonian terms, defined such that $h_k$ acts in the window $[k-l,k+l]$. Further, by our normalization $\langle h_k \rangle = 0$, so together with the restriction on the range of the $j$'s, we find each term in the sum for fixed $i_1 < \cdots < i_n$ is finite, and in fact, because of the prefactor, $O(1/L^n)$.

Consider the subsum of the $i$'s where at least one $i$ is within a distance $4l+2$ from another $i$. There are $O(L^{n-1})$ terms in this sum, and since each is $O(1/L^n)$ (using the bounded norm assumption on $h_k$), this part of the sum is $O(1/L)$ and in particular goes to zero as $L \to \infty$. Neglecting these terms allows us to always split the expectation value in $\Delta_n$ as a product of terms evaluated near each $i_m$. We can write it as
\[\Delta_n = -{L \choose n}^{-1} \sum_{\substack{i_1 < \cdots < i_n \\ i_{m+1} - i_m > 4l+2}} E(i_1,\ldots,i_n) + O(1/L)\]
where
\[E(i_1,\ldots,i_n) := \sum_{j_1 = i_1 - 2l-1}^{i_1 + 2l+1} \cdots \sum_{j_n = i_n - 2l-1}^{i_n + 2l+1} \sum_k \langle X_{i_1} \cdots X_{i_n} h_k X_{j_1} \cdots X_{j_n}\rangle \\ 
= \sum_{j_1 = i_1 - 2l-1}^{i_1 + 2l+1} \cdots \sum_{j_n = i_n - 2l-1}^{i_n + 2l+1} \Bigg(\sum_{k = i_1-2l-1}^{i_1+2l+1} \langle X_{i_1} h_k X_{j_1} \rangle \langle X_{i_2} X_{j_2} \rangle \cdots \langle X_{i_n} X_{j_n} \rangle\\ + \sum_{k = i_2-2l-1}^{i_2+2l+1} \langle X_{i_1} X_{j_1}\rangle \langle X_{i_2} h_k X_{j_2} \rangle \langle X_{i_3} X_{j_3} \rangle \cdots \langle X_{i_n} X_{j_n} \rangle + \cdots \\
+ \sum_{k = i_n-2l-1}^{i_n+2l+1} \langle X_{i_1} X_{j_1} \rangle \cdots \langle X_{i_{n-1}} X_{j_{n-1}} \rangle \langle X_{i_n} h_k X_{j_n} \rangle\Bigg) \\
 = \sum_{m = 1}^n \sum_{j = i_m - 2l-1}^{i_m + 2l+1} \sum_{k = i_m - 2l-1}^{i_m + 2l+1} \langle X_{i_m} h_k X_j \rangle \\
 = \sum_{m = 1}^n E(i_m).\]
Here we see the energy difference receiving a local contribution from each flipped spin individually, when they are well-separated.

Now we have
\[\Delta_n = -{L \choose n}^{-1} \sum_{\substack{i_1 < \cdots < i_n \\ i_{p+1} - i_p > 4l+2}} \sum_{m = 1}^{m = n} E(i_m) + O(1/L) \\ 
= -{L \choose n}^{-1} \sum_{m = 1}^{m = n} \sum_{\substack{i_1 < \cdots < i_n \\ i_{p+1} - i_p > 4l+2}} E(i_m) + O(1/L)
\\ = -{L \choose n}^{-1} \sum_{m = 1}^{m = n} \frac{1}{n!} \sum_{\substack{i_1, \ldots, i_n \\ |i_p - i_q| > 4l+2}} E(i_m) + O(1/L) \\ 
= -{L \choose n}^{-1} \frac{n}{n!} \sum_{i_1} E(i_1) \sum_{\substack{i_2, \ldots, i_n \\ |i_p - i_q| > 4l+2}} 1 + O(1/L).\]

(In the third line, we removed the ordering on the $i$'s.) Note the sum 
\[S = \sum_{\substack{i_2, \ldots, i_n \\ |i_p - i_q| > 4l+2}} 1\]
is independent of $i_1$. Then, using
\[\Delta_1 = -\frac{1}{L} \sum_i E(i)\]
we can do the sum over $i_1$ to obtain
\[\Delta_n = L n \Delta_1 {L \choose n}^{-1} \frac{1}{n!} S + O(1/L).\]
We can compute
\[S = \sum_{\substack{i_2, \ldots, i_n \\ |i_p - i_q| > 4l+2 \\ i_p > 4l+2}} 1 = (L-4l-2)(L-8l-4) \cdots (L-(n-1)(4l+2)) = L^{n-1} + O(L^{n-2}).\]
Note also,
\[{L \choose n}^{-1} \frac{1}{n!} = \frac{1}{L^n + O(L^{n-1})} = \frac{1}{L^n} + O(1/L^{n+1}).\]
Thus we have the result
\[\Delta_n = n \Delta_1 + O(1/L).\]

\section{Sketch of a generalization}

These arguments work in a slightly more general setting, which we will sketch. Let $\ket{0}$ be some fixed state in a one-dimensional $N$-site Hilbert space. We assume there is a length $\xi$ such that $\ket{0}$ obeys exponential cluster decomposition with length $\xi$, meaning for all local operators $\cO_i$, $\cO_j'$, there is a constant $K$ such that
\[|\langle \cO_i \cO_j' \rangle - \langle \cO_i \rangle \langle \cO_j' \rangle| < K e^{-|i-j|/\xi},\]
where we use the shorthand $\langle - \rangle$ for expectation values in $\ket{0}$.

Let $S_i$ be a local operator which behaves like a $U(1)$ charged operator, in particular, that the expectation value of a product of $S_i$'s and $S_j^\dagger$'s is zero unless there are an equal number of both. We call this the charge-balance property. It may be possible to relax this assumption, but the arguments don't seem to easily generalize to this case. (Note that the Hamiltonian is not assumed to be symmetric in any way.)

We consider the states
\[|W_n\rangle := {L \choose n}^{-1/2} C^{-1/2} \sum_{i_1 < \cdots < i_n} S_{i_1} \cdots S_{i_n} \ket{0},\]
where $C > 0$ is a normalization factor, defined as
\[C = {L \choose n}^{-1} \sum_{i_1 < \cdots < i_n} \sum_{j_1 < \cdots < j_n} \langle S^\dagger_{i_1} \cdots S^\dagger_{i_n} S_{j_1} \cdots S_{j_n} \rangle.\]
In this expression, as well as in the energy
\[\Delta_n = -{L \choose n}^{-1} C^{-1} \sum_{i_1 < \cdots < i_n} \sum_{j_1 < \cdots < j_n} \sum_k \langle S_{i_1}^\dagger \cdots S_{i_n}^\dagger h_k S_{j_1} \cdots S_{j_n}\rangle,\]
for fixed $i_1 < \cdots < i_n$, the sum over the $j$'s and $k$ is finite. When other operator insertions are far from these, by cluster decomposition, the expectation value becomes exponentially close to the product. However, by the charge balance property, these expectation values are exponentially small unless each $S_i^\dagger$ is near either an $S_j$ or $h_k$, and vice versa for each $S_j$. The summand is thus exponentially small unless $S_{i_m}^\dagger$ is near $S_{j_m}$ for each $m$. So we can restrict the sum to those $j_m$ which are within some fixed $l'$ of $i_m$, while incurring an error exponentially small in $l'$ and not growing with $N$. By taking $l'$ large enough we can thus make these errors arbitrarily small. The rest of the argument then proceeds as in the proof of proposition 1. Likewise we can consider the matrix elements as in proposition 2

The product state $\ket{0}$ we considered above satisfies these properties, with $\xi = 1$, $K = \|\cO_i\| \|\cO_j'\|$, $S_i = S_i^+ = (X_i + i Y_i)/2$, and even enjoys the $U(1)$ symmetry generated by $\sum_i Z_i$, under which $S_i$ carries charge 1. The states $\ket{n}$ above are the same ones we studied previously. There is a simple generalization to spin-$S$, with $S_i = (S_i^+)^m$, for which $\ket{n}$ are the so-called Dicke states \cite{Dicke}.

\section{Proof of Corollary~\ref{cor:momentumguys}}
\label{app:proofofcor}

\begin{proof}
    Let $H=\sum_i h_i$ be the local Hamiltonian (as defined previously) already in the normal form, i.e. $h_i=:h_i:$, with range $l$. Assume the $\ket{W_1}$ state is a ground state, and define momentum boost states $\ket{W_1^{(m)}}\equiv U_m\ket{W}$ with
    \[U_m\equiv \exp\left(i\frac{2\pi}{L}m\sum_{i}x_i\hat{n}_i\right)\]
    $m\in\mathbb{Z}$ fixed for all $L$. These states are eigenstates of the translation operator $T$
    \begin{align}
        T\ket{W_1^{(m)}}=e^{-i\frac{2\pi}{L}m}\ket{W_1^{(m)}}\quad,
        \label{eq:Teigenstate}
    \end{align}
    due to the relation $TU_mT^\dag=e^{-i\frac{2\pi\hat{N}}{L}m}U_m$.
    We are interested in the energy difference between $\ket{W_1^{(m)}}$ and $\ket{W_1}$:
    \begin{align}
        \Delta_W^{(m)}=\langle W_1^{(m)}|H\ket{W_1^{(m)}}-\langle W_1|H\ket{W_1}\quad,
    \end{align}
    which we will show behaves as $O(1/L^2)$ as $L \to \infty$. We may expand 
    \begin{align}
        \langle W_1^{(m)}|H\ket{W_1^{(m)}}& = \sum_i\langle W_1^{(m)}|h_i\ket{W_1^{(m)}},\nonumber\\
        &=\sum_i\langle W_1^{(m)}|T^{i-l}T^{-i+l}h_iT^{i-l}T^{-i+l}\ket{W_1^{(m)}},\nonumber\\
        &=\sum_i\langle W_1^{(m)}|T^{-i+l}h_iT^{i-l}\ket{W_1^{(m)}},\nonumber\\
        &=\sum_i\langle W_1^{(m)}|\tilde{h}_l^{(i)}\ket{W_1^{(m)}},
        \label{eq:kman}
    \end{align}
    where we have used Eq.~\ref{eq:Teigenstate} to go from line two to three, and have defined $\tilde{h}_l^{(i)}\equiv T^{-i+l}h_iT^{i-l}$ which is an operator centered at the $l$ affecting at most qubits sites $l$ away, i.e. the range $[0,2l]$. With this observation in mind, let us expand the equation for a specific $\tilde{h}_l^{(i)}$
    \begin{align}
        \langle W_1^{(m)}|h_l^{(i)}\ket{W_1^{(m)}}&=\langle W_1|e^{-i\frac{2\pi}{L}m\sum_{j}x_j\hat{n}_j}\tilde{h}_l^{(i)}e^{i\frac{2\pi}{L}m\sum_{k}x_k\hat{n}_k}\ket{W_1},\nonumber\\
        &=\langle W_1|\tilde{h}_l^{(i)}\ket{W_1}+i\frac{2\pi m}{L}\langle W_1|\left[\tilde{h}_l^{(i)},\sum_{j=0}^{2l}x_j\hat{n}_j\right]\ket{W_1}+\left(\frac{2\pi m}{L}\right)^2\langle W_1|\hat{h}^{(i)}_l\ket{W_1}+...\nonumber\\
        &=\langle W_1|\tilde{h}_l^{(i)}\ket{W_1}+\left(\frac{2\pi m}{L}\right)^2\langle W_1|\hat{h}^{(i)}_l\ket{W_1}+...\nonumber\\
        &=\langle W_1|\tilde{h}_l^{(i)}\ket{W_1}+O(1/L^3)\quad,
        \label{eq:manipkexcite}
    \end{align}
    where
    \begin{align}
        &\hat{h}^{(i)}_l\equiv\sum_{j,k=0}^{2l}x_j x_k\hat{n}_j \tilde{h}_l^{(i)}\hat{n}_k +\frac{1}{2}\left[\left(\sum_{j=0}^{2l}x_j\hat{n}_j \right)^2,\tilde{h}_l^{(i)}\right]\quad,
    \end{align}
    From line one to two in Eq.~\ref{eq:manipkexcite} we expanded the exponential in $1/L$ since $\sum_{j=0}^{2l}x_j\hat{n}_j\sim O(1)$ and terms outside of range $l$ do not contribute since $\tilde{h}_l^{(i)}$ is normal ordered and $W_1$ has no overlap between regions far away (using the same logic as for Proposition 2). From line two to three we use the fact that $\frac{2\pi m}{L}$ term must vanish since otherwise $\ket{W}$ would not be lowest ground state as we could choose either positive or negative $m$ to create a lower energy state~\cite{harukibloch}. The last line is derived by seeing that expectation value of normal ordered local operator $\langle W_1|\hat{h}^{(i)}_l\ket{W_1}$ can at most contribute $O(1/L)$ since terms where the $1$ is far-away from the range cannot contribute.

    Putting all the pieces together from Eqs.~\ref{eq:kman} and \ref{eq:manipkexcite} we can deduce that
    \begin{align}
        \Delta_W^{(m)}=O(1/L^2)\quad,\nonumber
    \end{align}
    since each $\tilde{h}_l^{(i)}$ contributes $O(1/L^3)$ and $\hat{h}_l^{(i)}$ is bounded so the summation over $i$ cannot scale as $O(L)$.

\end{proof}

\section{Fun facts}

In this section we present some fun facts and example Hamiltonians for readers to gain further intuition of the $\ket{W_1}$ type states.

\begin{cor}
Any state $\ket{\psi}$ related to the $\ket{W_1}$ state via a local unitary, or equivalently finite-depth quantum circuit (FDQC), that does not depend on system length, i.e. an adiabatic evolution by a Hamiltonian whose coefficients do not explicitly depend on length, also cannot be the unique ground state of a local Hamiltonian.

\end{cor}
\begin{proof}
    Let $\ket{\tilde{W}_1}=U_{\mathrm{FD}}\ket{W_1}$ where $U_{\mathrm{FD}}$ is a finite-depth quantum circuit. If $\ket{\tilde{W}_1}$ is the unique ground state of local Hamiltonian $\tilde{H}$, then $\ket{W_1}$ would be the unique ground state of local Hamiltonian $H=U_{\mathrm{FD}}\tilde{H}U_{\mathrm{FD}}^\dag$ since unitaries do not change the spectrum of the Hamiltonian. However since this cannot be true it follows that $\ket{\tilde{W}_1}$ cannot be the unique ground state of a local Hamiltonian.
\end{proof}
However this does not immediately include momentum boosted states $e^{i\frac{2\pi m}{L}\sum_{x}\hat{n}_xx} \ket{W_1}$ since the FDQC $e^{i\frac{2\pi m}{L}\sum_{x}\hat{n}_xx}$ has length dependent parameters when $m\in O(1)$. Here the statement may be modified to only hold if the parent Hamiltonian of the boosted state is related to a local Hamiltonian with length-independent parameters after the FDQC is applied - this immediately covers $U(1)$ and translation symmetric Hamiltonians that are related via large gauge transformations to a local Hamiltonian (length-independent) where $\ket{W_1}$ is the ground state such as in Eq.~\ref{eq:simpleHam}.

Following on the same vein as in Corollary~\ref{cor:momentumguys} we show the following statement
\begin{cor}
If $\ket{W_1}$ is a ground state of a local gapped $U(1)$ and translation invariant Hamiltonian then all momentum boosted states of $\ket{W_1}$, i.e. $e^{i\frac{2\pi m}{L}\sum_{x}\hat{n}_xx} \ket{W_1}$ $\forall m\in \mathbb{Z}$, are also ground states. Moreover, this automatically implies that the ground state degeneracy must be extensive with system size.
\end{cor}

\begin{proof}
    Let the system be described by a gapped, local Hamiltonian $H$ that commutes with both $U(1)$ and translation symmetry. This implies that there exists a basis that simultaneously diagonalises all three operators. For states with total particle number $\hat{N}=1$ this basis is spanned by the $\ket{W_1}$ state and its momentum boosts $\ket{m+1}=U_m\ket{W_1}$, where $U_m=e^{i\frac{2\pi m}{L}\sum_{x}\hat{n}_xx} $ and $m\in\{1,...,L-1\}$.

     Since these states are also necessarily eigenstates of $H$, they are either above or below the spectral gap. By Corollary~\ref{cor:momentumguys} we know that the boosts in the vicinity, i.e. when $m\in O(1)$, of $\ket{W_1}$, have an energy expectation value of $O(1/L^2)$. To satisfy this requirement, as well as being an eigenstate of $H$, these boosts must be below the spectral gap, i.e. elements of the ground state manifold.

     We may use the same arguments as in Corollary~\ref{cor:momentumguys} to show that for a ground states $\ket{m+1}$ the state $\ket{m+2}$ must be $O(1/L^2)$ in energy difference above the ground state. Due to the gap and symmetries, $\ket{m+2}$ must then also be in the ground state manifold. By doing this reiteratively for all $m\in\{1,...,L-1\}$ it follows that they are all in the ground state manifold. Since this is an extensive amount of states, the ground state manifold must always be extensive in size.
  
\end{proof}

\textit{Example Hamiltonians}. Some intuition of the corollaries can be gained by considering the following example of an (unstably) gapped Hamiltonian
\begin{align}
    H=\lambda\sum_i n_i n_{i+1}\quad,
\end{align}
where $n_i=c_i^\dag c_i$. Here $\ket{W_1}$ and momentum boosts of $\ket{W_1}$ are part of the extensive ground state manifold, while states with two neighbouring one's such as $\ket{110...0}$ are gapped excitations. The ground state degeneracy is unstable and can be lifted by a simple chemical potential term.

An interesting gapless Hamiltonian on an open spin chain with the $W$-state as a ground state is
\[H_L = -\sum_{j = 1}^{L-1} S^+_j S^-_{j+1} + hc.  - \frac12 Z_1 - \sum_{j=2}^{L-1} Z_j - \frac12 Z_L.\]
Its ground states are $\ket{0}$ and $\ket{W_1}$.

\end{document}